\documentclass[UKenglish]{llncs}

\usepackage{amsmath}
\usepackage{amsfonts}
\usepackage{amssymb}
\usepackage{color}
\IfFileExists{pdfsync.sty}{\usepackage{pdfsync}}{}

\newtheorem{define}{Definition}

\newcommand\motnouv[1]{{\emph{#1}}}

\newcommand\A{\mathcal{A}}
\newcommand\B{\mathfrak{B}}
\newcommand\C{\mathfrak{C}}

\newcommand\SSS{\mathcal{S}}
\newcommand\f{\mathfrak{F}}
\newcommand\ttt{\mathfrak{t}}
\newcommand\RL{\mathfrak{R}}

\newcommand\III{\mathcal{ND}}
\newcommand\D{\mathcal{D}}
\newcommand\F{\mathcal{F}}
\newcommand\T{\mathcal{T}}

\newcommand\N{\mathbb{N}}

\newcommand\id{identifier}

\newcommand\idd{identifiers}

\begin{document}

\title{Homonym Population Protocols}
\titlerunning{Homonym Population Protocols}




\newcommand\at{}

\author{Olivier Bournez
  \and Johanne Cohen 
\and Mika\"el Rabie
}
%
\institute{O. Bournez \at LIX, Ecole Polytechnique,  91128 Palaiseau Cedex, France
\email{bournez@lix.polytechnique.fr}
\and
J. Cohen \at
Universit\'{e} de Paris-Sud,
LRI, B\^{a}timent 425,\\
F-91405 Orsay Cedex, France\\\email{Johanne.Cohen@lri.fr}
\and M. Rabie \at LIX, Ecole Polytechnique,  91128 Palaiseau Cedex,
France
and Université Paris-Dauphine, 116 avenue du Président Wilson 93100 Montreuil Sous Bois France
\email{mikael.rabie@lix.polytechnique.fr}
}

\date{Received: \today/ Accepted: date}

\maketitle

\begin{abstract}
The population protocol model was introduced by Angluin \emph{et  al.}
as a model of passively mobile anonymous finite-state agents. This model  computes a
  predicate on the multiset of their inputs via interactions by
  pairs. The original population protocol  model has been proved to compute only semi-linear
  predicates and has been recently extended in various ways.

In the community protocol model by Guerraoui and Ruppert, agents have unique identifiers but
  may only store a finite number of the identifiers they already heard
  about.  The  community protocol model    provides the power of a
  Turing machine with a $O(n\log n)$ space.
   We consider variations on the two above models and we obtain a whole
  landscape that covers and  extends  already known results.

Namely, by  considering the case of homonyms, that is to say the case when several agents may share the same
  identifier, we provide a hierarchy that goes from the case of no
  identifier   (population protocol    model)
  to the case of unique identifiers  (community protocol
  model). We obtain in particular 
  that any Turing Machine on space
  $O(\log^{O(1)} n)$
  can be simulated with at least $O(\log^{O(1)} n)$
  identifiers, a result filling a gap left open in all previous studies.

Our results also extend and revisit in particular the hierarchy
  provided by Chatzigiannakis \emph{et al.} on population protocols
  carrying Turing Machines on limited space, solving the problem of
  the gap left by this work between per-agent space $o(\log \log n)$ (proved to be
  equivalent to population protocols) and $O(\log n)$ (proved to be
  equivalent to Turing machines).
       
%

\end{abstract}

\section{Introduction}

Angluin \emph{et al.}  \cite{AspnesADFP2004} proposed a  
distributed computation model called {\em population protocols}. It is a minimal model that aims at modeling large
sensor networks with resource-limited anonymous mobile agents. The
mobility of the agents is assumed to be unpredictable (but has to respect some
fair scheduler) and 
pairs of agents can
exchange state information 
when they are close
together.  

The population protocol model is also considered as a computational model, in
particular computing predicates: Given some input configuration, the
agents have to decide whether it satisfies the predicate.
The population of agents has to eventually stabilize to a
configuration in which every agent is in an accepting state or a
rejecting one. This should not depend on the size of the population,
 i.e for any
size of input configuration. 

The seminal work of Angluin \emph{et
al.} \cite{angluin2007cpp,AspnesADFP2004} proved that predicates
computed by population protocols are precisely 
first-order formula in Presburger arithmetic
(equivalent to a semilinear set).  Subsets definable
in this way are rather restricted, as multiplication for example is not
expressible in Presburger
arithmetic. 
%
Several variants
of the original model have been 
studied in order to strengthen the population protocol model with
additional realistic and implementable assumptions, in order to improve the
computational power. This includes natural restrictions like modifying the assumptions between
agent's interactions (one-way communications \cite{angluin2007cpp},
particular interaction graphs \cite{AngluinACFJP2005}).  This also
includes the 
Probabilistic Population Protocol model that makes a random scheduling
assumption for interactions~\cite{AspnesADFP2004}.  Also fault tolerance have been taken account  for population protocols
\cite{Delporte-GalletFGR06}, including the self-stabilizing
solutions \cite{OPODIS}.  We refer to
\cite{PopProtocolsEATCS,ChatzigiannakisMS10} for a 
survey.

Among many variants of population protocols, the \emph{passively mobile
(logarithmic space) machine model} introduced by Chatzigiannakis \emph{et
  al.}  \cite{chatzigiannakis2011passively} generalizes  the population protocol model where finite state
agents are replaced by agents that correspond to arbitrary Turing
machines with $O(S (n))$ space per-agent, where $n$ is the number
of agents.  An exact characterization
\cite{chatzigiannakis2011passively} of computable predicates is given: this model can compute all symmetric predicates in $NSPACE(n S (n))$
as long as $S (n) = \Omega (\log n)$.  Chatzigiannakis \emph{et al.}
establish that  the
model with a space in agent in  $S(n)=o(\log \log n)$ is equivalent to population protocols, i.e. to the case
$S(n)=O(1)$. 

In parallel, \emph{community protocols} introduced by Guerraoui and
Ruppert~\cite{guerraoui2009names} are closer to the original
population protocol model, assuming \textit{a priori} agents with
individual very restricted computational capabilities.  In this model,
each agent has a unique
identifier and  can only store $O(1)$ other agent
identifiers, and only identifiers from agents that it met.  Guerraoui
and Ruppert~\cite{guerraoui2009names} using results about the
so-called storage modification machines \cite{schonhage1980storage},   proved
that such protocols 
simulate  Turing machines:  predicates computed by this
model with $n$ agents are precisely the predicates in $NSPACE(n \log n)$.

This work aims at obtaining a whole landscape that covers and extends
already known results about \emph{community protocol} models and
population protocols carrying Turing Machines on limited space.
  
First,   
 we drop the hypothesis of unique identifiers. That is to say,  agents may have homonyms. We obtain a hierarchy that goes from the case of no
  identifier (i.e. population protocol
  model) to the case of unique identifiers (i.e. community protocol
  model). In what follows, $f(n)$ denotes the number of distinct available
  identifiers on a population with $n$ agents.  Notice that the idea
  of having less identifiers than agents, that is to say of having ``homonyms'',  has already been considered
  in other contexts or with not closely related problematics 
  \cite{Delporte-Gallet:2011,Delporte-Gallet:2012,ArevaloAIJR12,Guiseppe2013}.

Second, our results also extend the \emph{passively mobile
machine model}. In particular,
Chatzigiannakis \emph{et
  al.}  \cite{chatzigiannakis2011passively}  solved the cases
$S(n)=o(\log \log n)$ (equivalent to population protocols) and
$S(n)=O(\log n)$ (equivalent to Turing machines).
We provide a characterization for the case $S(n)=O(\log \log n)$:
the model is equivalent to
  $\bigcup_{k \in \N}SNSPACE(\log^k n)$ (see Table \ref{tab:Chatzigiannakis}).


To sum up, Tables~\ref{tab:Guerraoui}
and~\ref{tab:Chatzigiannakis} summarize  our results. 
$MNSPACE(S(n))$ (respectively:  $SMNSPACE(S(n))$) is the set of
$f$-symmetric\footnote{These classes are defined in Section~\ref{sec:def}.} (resp. also stable under
the permutation of the input multisets) 
languages recognized by non deterministic Turing
machines on space $O(S(n))$. 
   \begin{table}
\vspace{-0.3cm}
\begin{center}
\begin{tabular}{c|c}
$f(n)$ identifiers   &  Computational power\\\hline
$O(1)$ & Semilinear Sets \\
& \cite{angluin2007cpp,AspnesADFP2004} \\  \hline
$\Theta(\log^r n)$ &
$\bigcup_{k \in \mathbb{N}} MNSPACE\left(\log^k n
\right)$ \\
with $r\in\mathbb{R}_{>0}$&   Theorem \ref{th:main} \\ 
\hline
$\Theta(n^\epsilon)$&  $MNSPACE(n\log n)$ \\
with $\epsilon>0$& Theorem \ref{th:sqrt} \\
\hline 
$n$ &  $NSPACE(n\log n)$\\
& \cite{guerraoui2009names} \\
\end{tabular}
\end{center}
\caption{Homonym population
protocols with $n$ agents and $f(n)$ distinct identifiers.}
\label{tab:Guerraoui}

    \end{table}

\begin{table}
\vspace{-1.5cm}
\begin{center}
\begin{tabular}{c|c}
 Space per agent $S(n)$ & Computational power \\\hline
 $O(1)$  & Semilinear Sets  \\
                      &  \cite{angluin2007cpp,AspnesADFP2004}  \\  \hline
  $o(\log\log n)$    &  Semilinear Sets  \\
                                             &
              \cite{chatzigiannakis2011passively} \\ \hline
   $\Theta(\log\log n)$    &  $\bigcup_{k \in \mathbb{N}}SNSPACE(\log^k n)$  \\
                                     &  Theorem \ref{th:grec}\\ \hline
   $\Omega(\log n$) & $SNSPACE(nS(n))$ \\
   &    \cite{chatzigiannakis2011passively} \\
\end{tabular}
\end{center}
      \caption{Passively mobile
machine model~\cite{chatzigiannakis2011passively}
with $n$ agents and space $S(n)$ per agent.}
\label{tab:Chatzigiannakis}

 \end{table}

The document is organized as follows. Section~\ref{sec:2} introduces the formal definitions of the different
models and main known results. 
Section~3 introduces the definition of a new class of Turing Machines that will
help for the characterization of our models.
Section~\ref{sec:3p} is devoted to the case where there is
a polylogarithmic number of different identifiers in the population.
Section~\ref{resthirera} deals with the case where the population's number of identifiers is constant,
$o(\log\log n)$ and $\omega(n^{\epsilon})$. 
Section~\ref{sec:4demi} treats the case $S(n)=O(\log\log n)$ in the
passively mobile machine model~\cite{chatzigiannakis2011passively} (see Table 2).
Finally Section~\ref{sec:5} is  a summary of our results with some open questions.

\section{Models}
 \label{sec:2}

Population protocols have been mostly considered up to now as computing predicates:
one considers protocols such that starting from some initial
configuration, any fair sequence of pairwise interactions must
eventually lead to a state where all agents agree and either accept
or reject. A protocol computes a predicate if it accepts the inputs
verified by this predicate (and refuses the others).
%
Algorithms are assumed to be uniform: the protocol description must be 
independent of the number $n$ of the agents.

The models we consider are variations of the {\em community
  protocol} model~\cite{guerraoui2009names}. This latter model is in
turn   considered as an extension of the population protocols.
%
In all these models, 
a collection of
agents is considered. Each agent has a
finite number of possible states and an input value, that
determines  its initial state. Evolution of states of agents is the
product of pairwise interactions between agents: when two agents meet,
they exchange information about their states and simultaneously update
their own state according to a joint transition function, which
corresponds to the algorithm of the protocol. The precise
sequence of agents involved under the pairwise interactions is 
under the control of any fair scheduler.
The considered
notion of fairness for population protocols states that every 
configuration that can be reached infinitely often is eventually
reached.


%
%
%

In order to avoid multiplication of names, we will write community protocols for the
 model introduced by Guerraoui and Ruppert \cite{guerraoui2009names}, and homonym population protocols
for our version. 
The main difference between the two is the following: 
Let $U$ be the infinite set containing the possible
identifiers. 
We assume that the possible identifier set $U$ is not 
arbitrary: we assume that $U \subset \N$. 
We also assume these  identifiers are not necessarily
unique: several agents may have the same  identifier.
In a population of size $n$, we suppose that there are $f(n)$ distinct identifiers. 


More formally, a  community  protocol / homonym population protocol
algorithm is then specified by:
\begin{enumerate}
\item an infinite set $U$ of the possible identifiers.
In the Homonym case, $U=\mathbb{N}$.
\item a function $f: \N \to \N$ mapping the size of the population to the 
number of \idd{}  appearing in this population. In the community
protocol case, $f$ is identity. 
\item a finite set $B$ of possible internal states;
\item an integer $d \ge 0$ representing the number of  {identifier}s that can be
  remembered by an agent;
\item some input alphabet $\Sigma$ and some output alphabet $Y$;
\item an input map $\iota : \Sigma\rightarrow B$ and an output map $\omega: B\rightarrow Y$;
\item a transition function  $\delta:$ $Q^2\rightarrow Q^2$, with
  $Q=B\times U \times (U
  \cup \{\_\})^d$.
\end{enumerate}

\begin{remark}
For sake of clarity,  $\delta$ is a function, but this could be a relation as in the community protocol definition
\cite{guerraoui2009names}, without changing our results. 
\end{remark}

\begin{remark} Unlike in the population protocol model, each agent's
  state is given by its  identifier and   $d$ identifiers stored.  
If any of the $d$ slots is not currently storing an identifier, it
contains the \emph{null}  identifier $\_ \not\in U$.  In other words,
$Q=B\times U \times (U
  \cup \{\_\})^d$ is the set
of possible agent states.  
\end{remark}

The transition function $\delta$  indicates the
result (state shift) of a pairwise interaction: when agents in respective state
$q_1$ and $q_2$ meet, they move to respectively state $q'_1$ and
$q'_2$ whenever $\delta(q_1,q_2)=(q'_1,q'_2)$.

As in the community model \cite{guerraoui2009names},  agents store only  identifiers
they have learned from other agents (otherwise, they could be used as
an external way of storing arbitrary information and this could be used as
a space for computation in a non interesting and trivial way): 
if $\delta(q_1,q_2)=(q'_1,q'_2)$, and $id$ appears in
  $q'_1,q'_2$ then $id$ must appear in $q_1$ or in $q_2$.

The identifiers of agents
are chosen by some adversary, and not under the control of the
program. 

We add two hypothesis to the community model \cite{guerraoui2009names}:
agents need to know when an identifier is equal to $0$ and when
two identifiers are consecutive (i.e. $id_1=id_2+1$).
More formally, this is equivalent to say that 
%
%
  whenever $\delta(q_1,q_2)=(q'_1,q'_2)$, let  $u_1< u_2< \dots< u_k$ 
be the distinct  identifiers that appear in any of
  the four states $q_1,q_2,q'_1,q'_2$.
Let $v_1 < v_2 < \dots < v_k$ be identifiers such that  $u_1=0\Leftrightarrow v_1=0$ and $v_i+1=v_{i+1}\Leftrightarrow u_i+1=u_{i+1}$.
If $\rho(q)$ is the state obtained from $q$ by replacing all
  occurrences of each identifier $u_i$ by $v_i$, then we require that
  $\delta(\rho(q_1),\rho(q_2))=(\rho(q'_1),\rho(q'_2))$. 

 We also suppose that
every identifiers in $[0,f(n)-1]$ are present in the population. As we want to be
minimal,  those are the only hypothesis we make on identifiers in the following
sections. 



\begin{remark} \ 

\vspace{-0.3cm}
\begin{itemize}
\item This weakening of the community protocols does not change the computational
power in the case where all agents have distinct identifiers.
\item 
Our purpose is to establish results with minimal hypothesis. Our
results work  when identifiers are consecutive integers, say
$\{0,1,2,\dots,f(n)-1\}$.  
This may be thought as a restriction. This is why we weaken to the
above hypothesis, which seems to be the minimal hypothesis to make our
proofs and constructions correct. 

We conjecture that without the possibility to know if an identifier is
the successor of another one, the model is far too weak. Without this
assumption, our first protocol (in  Proposition~$1$) does not work.
\item
Notice that 
knowing whether an identifier is equal to $0$ is not essential, but eases
the explanation of our counting protocol of Proposition~$1$.
\end{itemize}
\end{remark}

  From now on, an agent in state $q$ with initial identifier $k$ and
  $L=k_1,\ldots ,k_d$ the list storing the $d$ identifiers is denoted
  by $q_{k,L}$ or $q_{k,k_1,\ldots ,k_d}$.  If the list $L$ is not
  relevant for the rule, we sometimes write $q_k$ where $k$ is the agent's identifier.

A \emph{configuration} $C$ of the algorithm then consists of a finite vector of
elements from $Q$:  it is a set of $n$ agent's states.
 An \emph{input of size $n\ge2$} is given by $f(n)$ non empty multisets $X_i$
over alphabet $\Sigma$.
 An initial configuration for $n$ agents is a vector in $Q^n$
of the form $((\iota(x_j),i,\_,\dots,\_))_{0\le i < f(n), 1\le j\le \left|X_i\right|}$
where $x_j$ is the $j$th element of $X_i$. In other words, every agent $x_j$
starts in a basic state encoding $\iota(x_j)$, its associated
identifier  and
no other identifier stored in its $d$ slots (since it met no other agent). 


 If $C=(q^{(1)},q^{(2)},\dots,q^{(n)})$ and $C'=(p^{(1)},p^{(2)},\dots,p^{(n)})$ are two
configurations, then we say that $C \to C'$ ($C'$ is reachable from $C$
in a unique step) if  $C'$ is obtained by executing only one step of the transition function. In other words,  there are $i \neq j$ such that
$\delta(q^{(i)},q^{(j)})=(p^{(i)},p^{(j)})$ and $p^{(k)}=q^{(k)}$ for all $k$ different
from $i$ and $j$. An execution is a sequence of configurations $C_0,
C_1, \dots,$ such that $C_0$ is an initial configuration, and $C_i
\to C_{i+1}$ for all $i$. An execution is fair if for each
configuration $C$ that appears infinitely often and for each $C'$
such that $C \to C'$, $C'$ appears infinitely often.

\begin{example}[Leader Election]
  We want to design a protocol that performs a leader election, with the additional
  hypothesis that when the election has finished, all agents know the
  identifier of the leader.

  We adapt here a classical example of Population Protocol (Recall that in the  Population Protocol, this
  is not possible to store the identifier of the leader and hence this
  is not possible to compute the leader's identifier; Here this is
  possible as identifiers can be stored):
    Each agent with identifier $k$ starts with state $L_{k,\_}$,
  considering that the leader is an agent with identifier $k$. We want to design a protocol such that
  eventually at some time (i.e in a finite number of steps), there
  will be a unique agent in state $L_{k_0,k_0}$, where $k_0$ is the
  identifier of this unique agent, and all the other agents in state
  $N_{i,k_0}$ (where $i$ is its identifier).

 A protocol performing such a leader election is the following: 
  $f(n) =n$, $U=\N$, $B=\{L,N\}$, 
  $d=1$ (only the identifier of the current leader is stored), 
  $\Sigma=\{L\}$, $Y=\{True\}$, 
  $\iota(L)=L$, 
  $\omega(L)=\omega(N)=True$, $Q = B \times U \times (U\cup\{\_ \})$
and  $\delta$  such that the rules are: 
%
\begin{center}
\begin{tabular}{ l @{\hspace{0,2cm}} l @{$\rightarrow$} l
    @{\hspace{0,2cm}} l @{\hspace{1cm}} l }
  $L_{k,\_} $ & $ q$ & $ L_{k,k} $ & $  q$  & $ \forall k\in\N,\forall q\in Q$\\
  $L_{k,k} $ & $  L_{k',k'}$ & $ L_{k,k} $ & $  N_{k',k}$  & $ \forall k,k'$\\
$L_{k,k} $ & $  N _{i,k'} $ & $  L_{k,k} $ & $  N _{i,k}$ & $\forall k,k',i$\\
$  N _{i,k'} $ & $L_{k,k} $  &  $  N _{i,k}$  & $  L_{k,k} $ & $\forall k,k',i$\\
$N_{i,k} $ & $ N_{i',k'}$ & $N_{i,k} $ & $ N_{i',k'}$ & $\forall k,k',i,i'$\\
\end{tabular}
\end{center}

By the fairness assumption, this protocol will
reach a configuration where there is exactly
one agent in state $L_{k_0,k_0}$ for some identifier $k_0$.
Then, by fairness again, this protocol will reach the
final configuration such that only agent is in state $L_{k_0,k_0}$ and all the other agents are in state $N_{i,k_0}$ with ${i\ne k_0}$.

\end{example}

A configuration has an \motnouv{Interpretation} $y\in Y$ if, for each agent in the
population, its state $q$ is such that $\omega(q)=y$. If there are two agents in state
$q_1$ and $q_2$ such that $\omega(q_1)\ne\omega(q_2)$, then we say that the
configuration has \motnouv{no Interpretation}.
%
%
%
%
%
%
A protocol is said to \motnouv{compute the output} $y$ from an input $x$  if,
for each fair sequence $(C_i)_{i\in\mathbb{N}}$ starting from an
initial condition $C_0$ representing $x$, there exists $i$
such that, for each $j\ge i$, $C_j$ has the interpretation $y$.
The protocol is said to compute function $h$ if
it computes $y=h(x)$ for all inputs $x$.
 A predicate is a function $h$ 
whose range is 
$Y=\{0,1\}$. As usual, a predicate can also
be considered as a decision problem, and a decision problem is said to
be decided if its characteristic function is computed.



Observe that  population protocols 
\cite{angluin2007cpp,AspnesADFP2004} are the special case  of the
protocols considered here where $d=0$ and $f(n)=1$.
%
The following is known for the population protocol model~\cite{angluin2007cpp,AspnesADFP2004}: 

\begin{theorem}[Population Protocols \cite{angluin2007cpp}]
Any predicate over $\N^k$ that is first order definable in
Presburger's arithmetic can be computed by a population protocol.
Conversely, 
any predicate computed by a population protocol is a subset of
$\N^k$ first order definable in
Presburger's arithmetic.
\end{theorem}




For the community protocols, Guerraoui and Ruppert~\cite{guerraoui2009names} 
established     that computable predicates
are exactly those of $NSPACE(n\log n)$, i.e. those of  the class
  of languages recognized in non-deterministic space $n\log n$.

Notice that their
convention \cite{guerraoui2009names}  of input   requires that the input be distributed on agents ordered by
identifiers. 
%


\begin{theorem}[Community Protocols
\cite{guerraoui2009names}] \label{th:eux}
Community protocols can compute  any predicate in $NSPACE(n\log n)$. 
Conversely, any predicate computed by such a 
community protocol is in the class $NSPACE(n\log n)$.
\end{theorem}

Notice that Guerraoui and Ruppert \cite{guerraoui2009names}
established that this holds even with Byzantine agents, under some
rather strong conditions.  We now determine what can be computed when
the number of identifiers $f(n)$ is smaller than $n$. This will be
done by first considering some basic protocols.

\section{Definitions}
\label{sec:3}

Our  main aim  is to determine exactly what can be computed
with homonym population protocols. We first need to introduce
Turing Machines that has inputs analog to homonym protocols. To perform it
we  will define the notion of \emph{$(f,n)$-symmetry} language,
and   the  class $MNSPACE$.


A multiset of elements from some set $\Sigma$ is also  seen as a
word over alphabet $\Sigma$: list the elements of the multiset in any
order, and consider the list as a word. Of course, any permutation of
this word corresponds to the same multiset. 

We focus on languages made of $f(n)$ multisets over alphabet
$\Sigma$. From the above remark, this is seen as words over alphabet
$\Sigma\cup\{\#\}$, and hence as $(f,n)$-Symmetric languages in the
following sense:

\begin{define}\label{defmnspace}
 Let $f:\mathbb{N}\to\mathbb{N}$ be a function. 
A Language $L$ over alphabet $\Sigma\cup\{\#\}$ is  \motnouv{$(f,n)$-Symmetric} if  and only if:
\begin{itemize}
\item \#$\not\in\Sigma$;
\item Words of $L$ are all of the form $w=x_1$\#$x_2$\#$\ldots$\#$x_{f(n)}$, with $|x_1|+|x_2|+\ldots+|x_{f(n)}|=n$ and
$\forall i$, $x_i\in\Sigma^+$;
\item If, $\forall i$, $x'_i$ is a permutation of $x_i$, and if $x_1$\#$x_2$\#$\ldots$\#$x_{f}\in L$, 
then \\ $x'_1$\#$x'_2$\#$\ldots$\#$x'_{f}\in L$;

\end{itemize}
\end{define}

The complexity class associated to   $(f,n)$-symmetric languages is :

\begin{define}[$MNSPACE(S(n))$]
Let $S$ be a function 
$\mathbb{N}\to \mathbb{N}$. The set of $(f,n)$-symmetric languages recognized by Non Deterministic
Turing Machines on  space $O(S(n))$ is the class   \motnouv{$MNSPACE(S(n),f(n))$}
 or  \motnouv{$MNSPACE(S(n))$} when $f$ is unambiguous.
\end{define}

\begin{remark} $NSPACE(S(n))=MNSPACE(S(n),n)$ since each multiset
must then contain exactly one element.
     $SNSPACE(S(n))=MNSPACE(S(n),1)$  as accepting a multiset
is exactly being stable under input permutation.
\end{remark}

Now, we will define a collection of languages called \emph{Included Language} in $MNSPACE(\log n)$. 

\begin{define}
Let $I$ be a positive integer and let $(X_i)_{i\le I}$ be a finite
sequence of multisets of elements from $\Sigma$.
$(X_i)_{i\le I}$ is a  \motnouv{included language}, 
if and only if, for all $i<I$, $\emptyset\ne X_{i+1}\subset X_{i}$.
\end{define}

This corresponds to finite sequences of non-empty multisets where each
multiset is included in the previous one. Using the above
representation trick (representing a multiset by a word up to
permutation) this can also be considered as a $(I,n)$-Symmetric
language.

\begin{proposition}
Any included language $(X_i)_{i\le I}$ is in $MNSPACE(\log n)$ where $\displaystyle n = \sum_{1 \le i \le I} |X_{i}|$.
\end{proposition}
\begin{proof}
Let $I$ be a positive integer and let $(X_i)_{i\le I}$ be an included language.
$|X_i|_s$ denotes the number of $s$ in a multiset $X_i$. 

To check whether  $X_{i+1}\subset X_i$ is equivalent to check, whether $|X_{i+1}|_s-|X_{i}|_s\le0$
for each $s\in\Sigma$. 
This can be performed on a space $O(\log n)$. 
In other words,  the language is accepted if and only if there do not exist $i$ and $s$ such that
$|X_{i+1}|_s-|X_{i}|_s\le0$.\hfill \qed

\end{proof}



\begin{proposition} Let $I$ be a positive integer and let $(X_i)_{i\le I}$ be an included language
and let $\displaystyle n = \sum_{1 \le i \le I} |X_{i}|$.

 $(X_i)_{i\le I}$ can be decided by a homonym population protocol of $n$ agents having $f(n) =I$ distinct identifiers.
\end{proposition}
\begin{proof}

The high-level description of the protocol is  described in order to check whether $X_{i+1}\subset X_i$ for every $i$.  Thus,  $X_{i+1}\subset X_i$ implies that each
agent in $X_{i+1}$ manages to "delete" an agent with the same input in
$X_i$.

In our protocol, each agent with an identifier $i>0$ looks for
an agent of identifier $i-1$ with the same input to "delete" it. 
More formally, consider:

\begin{itemize}
\item $B=\Sigma\times\{\III,\D\}\times\{\III,\D\}\times\{\T,\ttt,\F,\f\}$.
\item $d=0$.
\item $Y=\{True, False\}$.
\item $\forall s\in\Sigma,$ $\iota(s)=(s,\III,\III,\F)$.
\item $\forall s\in\Sigma$, $\forall a,b\in\{\III,\D\}^2$, $\omega(s,a,b,\T)=\omega(s,a,b,\ttt)=True$,
$\omega(s,a,b,\F)=\omega(s,a,b,\f)=False$.
\item $\delta$ is such that the non trivial rules are:
\end{itemize}
\begin{center}
\begin{tabular}{ r @{\hspace{0,2cm}} l @{$\rightarrow$} r @{\hspace{0,2cm}}  l l}
$(s,\III,b,\F)_0$ & $q$ & $(s,\D,b,\T)_0$ & $q$ & $\forall b,q$\\
$(s,a,\III,c)_i$ & $(s,\III,b,c')_{i+1}$ & $(s,a,\D,c)_i$ & $(s,\D,b,\T)_{i+1}$ & $\forall a,b,c,c'$\\
$(s,a,b,\F)_i$ & $(s',a',b',c')_j$ & $(s,a,b,\F)_i$ & $(s',a',b',\f)_j$ & $\forall a,b,a',b',c'$\\
$(s,a,b,\T)_i$ & $(s',a',b',\f)_j$ & $(s,a,b,\T)_i$ & $(s',a',b',\ttt)_j$ & $\forall a,b,a',b',c'$\\
\end{tabular}
\end{center}

The state of each agent is composed of four elements:
\begin{enumerate}
\item The first element corresponds to its input symbol.
\item The second is equal to $\III$ if the agent with $id$ identifier has not yet "deleted" an
agent with the same input and with $id-1$ identifier. It is equal to $\D$
if the deletion has already been performed.
\item The third element is equal to $\III$ if it has not been deleted yet by an agent
with the successor identifier. It is $\D$ as soon as the deletion has been performed.
\item The fourth element corresponds to the output. State $\F$ means that
the agent needs to perform a deletion. The agent knows that the input has to be $False$
as long as it has not deleted an agent with the previous identifier. State $\f$ means
the agent believes that at least one deletion needs to be performed.
State $\T$ means that the agent made its deletion and since did not meet agent needing
to perform a deletion. State $\ttt$ means that the agent believes that no deletions need to be done.
\end{enumerate}

Rule $1$ handles that agents with identifier $0$ does not need to do a deletion.
Rule 2 handles a deletion.

If we project on the fourth element, the stable configuration are
$\T^*\ttt^*$ and $\F^*\f^*$.

Rule 3 spreads the output $False$ to each agent. It can only come
from an agent still waiting for a deletion. If a deletion needs to be done
and cannot be, the output $False$ will be spread by fairness.

Rule 4 spreads the output $True$ from $\T$ to $\f$ agents.
If there are no longer deletions to do, there is at least one $\T$ in the population,
being from the agent that has performed the last deletion.
It will spread the output $True$ by fairness. \hfill \qed
\end{proof}



\section{A Polylogarithmic Number of Identifiers}\label{sec:3p}

In this section, the case where the population  contains  $f(n)\ge\log n$ distinct
identifiers is considered. First, a protocol is designed in order to compute   the size of the population. 
When the population reaches a stable configuration (i.e. agents will no longer be able
to fin another one that will change its state), the size 
will be encoded in binary  on $\log n$ agents. We will then show how to 
"read the input" and how to simulate a tape of length $\log n$. To perform that,
we will explain a process that ensures that the protocol will at some point
do exactly what is expected.

\subsection{Organization As a Chain}

The first step is to organize $f(n)$ agents in a chain: 
  We design a protocol that creates a chain containing all the existing identifiers such that an agent with identifier $Id_{k+1}$ is a successor of an agent with identifier $Id_{k}$.  This protocol consists of the execution of several leader Election protocols.
  The classical Leader Election protocol for population protocols distinguishes
one agent from all the others, usually by having all agents but one in state $N$,
the leader being in state $L$  by using the simple rule $L$ $L\to L$
$N$. Here, instead of having a unique agent different from all the
others, it distinguishes one agent for each distinct identifier, by 
the simple trick that two agents truly interact only if
they have the same identifier.
Here is the protocol:
\begin{itemize}
\item $B=\{L,N\}$.
\item $d=0$.
\item $\Sigma=L$, $Y=\{True\}$, $\iota(L)=L$ and $\omega(L)=\omega(N)=True$.
\item $\delta$ has only one difference from the usual Leader Election:
the first rule is split according to the identifiers.
\end{itemize}
\begin{center}
\begin{tabular}{ r @{\hspace{0,2cm}} l @{$\rightarrow$} r @{\hspace{0,2cm}}  l l}
$L_{id_a}$ & $L_{id_a}$ & $L_{id_a}$ & $N_{id_a}$ & \\
$L_{id_a}$ & $L_{id_b}$ & $L_{id_a}$ & $L_{id_b}$ & with $id_a\ne id_b$\\
\end{tabular}
\end{center}

At some point, there will be exactly one $L_{id}$ agent for each $id\le f(n)$.
We will denote the leader with identifier ${id}$ by $L_{id}$. Moreover, 
we focus on the leader with identifier $0$ denoted by  $L_0$ and called the $\motnouv{Leader}$.

For the remaining of this section, the \motnouv{Chain} will refer to these particular agents.
We will often see these agents as a tape of $f(n)$ symbols in $B$
sorted according to the identifiers.

\subsection{The Size of the Population}

In many classical population protocols, integers are encoded in unary.
Here, we will compute the size of the population 
%
%
on $\log n$ agents in binary.
The size of the population will be encoded in binary on a chain
(it will be possible as $f(n)\ge\log n$ in this part), using the
concept of chain of previous section. 


\begin{proposition}\label{compte}
When the population has $f(n)\ge\log n$ distinct identifiers,
there exists an homonym population protocol that computes the size $n$
of the population: 
At some point,  all agents are in a particular state $N$ except   $f(n)$ agents having distinct identifiers.
If we align these agents from the highest identifier
to the lowest one,  then they encode the size of the population $n$ written  in binary.
\end{proposition}


\begin{proof}
  In a high-level description, the protocol initializes all agents to a particular
  state $A$. This protocol will implicitly include the chain construction
(instead of being in state $L$, potential leaders will have their state
in $\{A,0,1,2\}$).
This protocol  counts the number of agents in state $A$. An agent in state
  $1$ (respectively 0, or 2) with identifier $k$ represents $2^k$
  (respectively 0, or $2^{k+1}$) agents counted. Interactions between
  agents are then built to update those counts. 

More formally, we have $B_{c}=\{A,0,1,2,N\}$ . The rules $\delta_{c}$ are the following:
\begin{center}
\begin{tabular}{ l @{\hspace{0,2cm}} l @{$\rightarrow$} l
    @{\hspace{0,2cm}} l @{\hspace{1cm}} l @{\hspace{2cm}}  l @{\hspace{0,2cm}} l
    @{\hspace{0,2cm}}l @{\hspace{0,2cm}} l @{\hspace{0,2cm}} l @{\hspace{1cm}} l }
$A_0$ & $q_k$ & $1_0$ & $q_k$ & $\forall q, k$  & $0_k$ & $1_k$ & $\rightarrow$ &$N_k$ & $1_k$ & $\forall k$\\
$A_k$ & $0_0$ & $0_k$ & $1_0$ & $\forall k\ge1$ & $1_k$ & $1_k$ & $\rightarrow$&$N_k$ & $2_k$ & $\forall k$\\
$A_k$ & $1_0$ & $0_k$ & $2_0$ & $\forall k\ge1$ & $0_k$ & $0_k$ & $\rightarrow$&$N_k$ & $0_k$ & $\forall k$ \\
$0_{k+1}$ & $2_k$ & $1_{k+1}$ & $0_k$ & $\forall k$ &       &            &  &            &            & \\
$1_{k+1}$ & $2_k$ &  $2_{k+1}$ & $0_k$ & $\forall k$&       &            &   &          &            &\\
\end{tabular}
\end{center}

This protocol is split into 3 steps. (i) At the beginning, all agents are in state $A$.  A state $A$ is
  transformed  into a state $1$, by adding $1$ to an agent of
  identifier $0$ corresponding to the $3$ first rules of the left column. (ii) The remaining rules of the left column correspond to summing together the counted agents, carrying on to the next identifier the $1$.
 (iii) Rules of the right column are here
to perform in parallel the chain protocol.
 
Let $v$ be the function over the states defined as follows for any $k$:
$v(A_k)=1$, $v(0_k)=v(N_k)=0$, $v(1_k)=2^k$, $v(2_k)=2^{k+1}$.  We
can notice that the sum of $v$ values over all the agents remains constant over the
rules. Thus the sum always equals the number of agents in the
population.

By fairness, all the agents in state $A$ will disappear, the chain will finish, and the agent in state $2_k$ will disappear.
Hence, the protocol writes the size of the population on the chain in binary. \hfill \qed
\end{proof}

\begin{remark}
The  previous counting protocol also works with $f(n)=\Omega(\log n)$.
Indeed, if for some $\alpha<1$ we have $f(n)\ge\alpha\log n$,
then, using a base $\lceil e^{1/\alpha}\rceil$ instead of a base 2 allows that $n$ can
be written on $f(n)$ digits.
\end{remark}

\begin{remark}   The  previous counting protocol  works if the population can detect that  an identifier is equal to $0$. 
This protocol can be adapted to a population 
with identifiers in $[a,a+f(n)-1]$. For this, agents store an identifier $Id_{m}$
corresponding to the minimal one they met (called here $Id_m$). An agent with identifier $Id$ and state $i\in\{0,1,2\}$
stores $i\cdot 2^{Id-Id_m}$. When it meets an identifier equals to $Id_m-1$,
it looks for a leader with identifier $Id-1$ to give it its stored integer.
\end{remark}

Once 
a chain is constructed, as above, it can be used to store numbers or
words. Thus it can be used as the tape of a Turing Machine. We will
often implicitly use this trick within the rest of the paper.  


\subsection{Resetting a Computation}

The computation of the size $n$ (encoded as above) is crucial for the following protocols.  
We call \emph{leader}  an agent not in state $N$ with identifier 0
from the previous protocol (or the chain protocol).

We will now provide a  \emph{Reset  protocol}. This protocol has the goal to reach a configuration
such that (i) the Counting protocol is finished, (ii) all agents except the leader are in state $R$, and (iii)
the leader {\it knows} when this configuration is reached (i.e. the leader is
on a state $\RL$ making him the belief the reset is done. The last time 
the protocol turn the leader's state into $\RL$, the reset will be done).

This protocol will then permit to launch the computation of some other protocols
with the guarantee that, by fairness, at some point, all agents were counted, and hence covered.

\begin{define}
A \motnouv{Reset Protocol} is a homonym protocol that guarantees to reach a configuration
where:
\begin{itemize}
\item The size of the population has been computed. In particular, the population has
a unique leader at this point.
\item There exists a mapping function $map:Q \to\{R,\RL\}$ such
  that leader's state is mapped to $\RL$ and 
all other agents' states are mapped to $R$. 
\end{itemize}

\end{define}
This configuration will be the beginning of the next computation step. All
agents will be ready, at this point, to start a next computation.
\begin{proposition}\label{prop:reset}
There exists a  Reset  Protocol.
\end{proposition}

\begin{proof}
 The high-level description  of the protocol relies on  starting back the reset protocol each time the leader sees that
  the counting protocol (of the previous proposition) has not finished yet.

The leader turns agents in state $R$. The chain lets the leader count the number of agents it turns into state $R$.
If it has turned the same number of agents that the number computed by the
Counting protocol, then it turns its state into $\RL$.
The protocol handles two counters in parallel. The first counter corresponds to the size of the population and the second counter
counts the number of agents that have been reset. We can notice that the leaders for each identifier will be the same on the first and second element,
and hence the chains will be the same for the two counters.

This mechanism is not as simple as it could look: the protocol uses the fairness
to be sure that at some point, it will have turned the right number of agents into state $R$.\\

More formally, the set of states is $$B_c\times B_c\times\{\A,\B,\C,D,D^p,D'^p,E,\RL,R,\SSS,W\},$$ where $B_c$ is the set
of states of the counting protocol (and $\delta_c$ its interaction function). The first element of the triplet
is for the counting protocol, the second for the counting of agents turned into
state $R$, and the third part is for the reset itself. At the beginning, all agents are in state $(A,0,\A)$. 
The second counting protocol will work a bit differently for the carry over.
The leader performs it by itself, walking through the chain, contrarily to the process
described in Proposition~\ref{compte}.

This protocol uses two identifiers slots. The first will be attached to the first counter:
the leader stores the greater identifier it heard about (each time it updates it,
we consider it has updated its state). The second is attached to the second counter.

Sometimes, rules will be of the form:
\begin{center}
\begin{tabular}{ r @{\hspace{0,2cm}} l @{$\rightarrow$} r @{\hspace{0,2cm}} l l }
$(q_1,q_2,q_3)$ & $(q_4,q_5,q_6)$ & $(q'_1,q'_2,q'_3)$ & $(q'_4,q'_5,q'_6)$. 
\end{tabular}
\end{center}
We will implicitly assume that $\delta_c(q_1,q_4)=(q'_1,q'_4)$.
If we write $q_1$ (resp $q_4$) instead of $q'_1$ (resp $q'_4$), then
it implies that $q'_1=q_1$ (resp $q'_4=q_4$).
We will note $q_{1,id}$ if the identifier $id$ attached to $q_1$ is relevant
(same for $q_2$).

We will describe $\delta$ according to the steps of the process of resetting:
\begin{enumerate}
\item First, the leader needs to know when the counting protocol evolved.
For this, as soon as an interaction occurred, agents not being a leader go into state $W$
to warn the leader to restart the reset. The leader $L_0$ then goes
into state $\A$ on its third element. (We can notice that in $\delta_c$,
there is a non trivial interaction with a leader if and only if the leader is the second element.
We can also notice that an agent with identifier $0$ in state $q_1\ne N$ is a leader.)
\end{enumerate}
\begin{center}
\begin{tabular}{ r @{\hspace{0,2cm}} l @{$\rightarrow$} r @{\hspace{0,2cm}} l l }
$(q_1,q_2,q_3)$ & $(q_4,q_5,q_6)$ & $(q'_1,q_2,W)$ & $(q'_4,q_5,q_6)$ & with $q'_1\ne q_1$ or $q'_4\ne q_4$\\
$(q_1,q_2,q_3)_0$ & $(q_4,q_5,W)$ & $(q'_1,q_2,\A)_0$ & $(q'_4,q_5,\SSS)$ & with $ q_1\ne N$\\
\end{tabular}
\end{center}

\begin{enumerate}
\item[2.] In  state $\A$, the leader knows that the counting protocol is not finished yet. 
It turns all agents into state $\SSS$. At some point (when the leader is the second
element of an interaction), it stops and goes to state $\B$.

The idea is that by fairness, if we repeat again and again this process,
at some point, the leader will manage to have all other agents turned into state $\SSS$.
\end{enumerate}
\begin{center}
\begin{tabular}{ r @{\hspace{0,2cm}} l @{$\rightarrow$} r @{\hspace{0,2cm}} l l }
$(q_1,q_2,\A)_0$ & $(q_4,q_5,q_6)$ & $(q_1,q_2,\A)_0$ & $(q_4,q_5,\SSS)$ & with $q_1\ne N$\\
$(q_1,q_2,q_3)$ & $(q_4,q_5,\A)_0$ & $(q_1,q_2,\SSS)$ & $(q_4,1_0,\B)_0$ & with $q_4\ne N$\\
\end{tabular}
\end{center}

\begin{enumerate}
\item[3.] In state $\B$, the leader clears the second chain
  corresponding to the swith econd counter of the reset agents.
This way, after the last change from the counting protocol,
we are sure that the chain will be cleared and will effectively
count all the agents.
The identifier attached to the first state gives to the leader the highest
identifier it saw. This way, it knows the last element's identifier in the chain.

The leader just keeps its own bit at $1$, as it needs to
count itself. After that, the leader goes to state $\C$.
\end{enumerate}
\begin{center}
\begin{tabular}{ r @{\hspace{0,2cm}} l @{$\rightarrow$} r @{\hspace{0,2cm}} l l }
$(q_{1_i},1_j,\B)_0$ & $(q_4,q_5,q_6)_{j+1}$ & $(q_1,1_{j+1},\B)_0$ &
                                                                      $(q_4,0,q_6)_{j+1}$
                                                                             &
                                                                               with $j<i$-1 \\  and $q_5\ne N$\\
$(q_{1_i+1},1_i,\B)_0$ & $(q_4,q_5,q_6)_{i+1}$ & $(q_1,1_\_,\C)_0$ & $(q_4,0,q_6)_{i+1}$ & with $q_5\ne N$\\
\end{tabular}
\end{center}

\begin{enumerate}
\item[4] In state $C$, the leader looks only for $S$ agents. When it finds one
it turns it into a $R$ and adds $1$ to the second counter (by going to state $D$).
If it meets another state, it goes back to state $A$, in order to try to
turn again all agents into $S$ and to reset the counter.
\end{enumerate}
\begin{center}
\begin{tabular}{ r @{\hspace{0,2cm}} l @{$\rightarrow$} r @{\hspace{0,2cm}} l l }
$(q_1,q_2,\C)_0$ & $(q_4,q_5,\SSS)$ & $(q_1,q_2,D)_0$ & $(q_4,q_5,R)$ &\\
$(q_1,q_2,\C)_0$ & $(q_4,q_5,q_6)$ & $(q_1,q_2,\A)_0$ & $(q_4,q_5,\SSS)$ & with $q_6\ne S$\\
\end{tabular}
\end{center}

\begin{enumerate}
\item[5] In state $D$, the leader increases the second counter by $1$.
If it has a carry, it goes in state $D^p$ as long at it is needed.
When the incrementation is over it goes back to state $D$ until it
reaches the end of the chain or find a difference.
If it finds a difference, and if it still has to propagate the carry, then it goes in state $D'^p$.

Here, we see why the agents need another slot of identifier: the leader needs to
remember what was the last bit he saw
(to identify easily the next one to find). The identifier on the first state
allows to know until which identifier
it has to compare the two counters.

If it reaches the last bit of the chain and the two counters are equal,
then the leader believes the reseting is over and goes to state $\RL$ (until
the counting protocol on the first element gets an update, if it happens).
If the counter is not yet equal, then it looks for another $\SSS$ to turn into a $R$.
\end{enumerate}
\begin{center}
{\small
\begin{tabular}{ r @{\hspace{0,1cm}} l @{$\rightarrow$} r @{\hspace{0,1cm}} l l }
$(0_i,0_\_,D)_0$ & $(q_4,q_5,q_6)$ & $(1_i,0_\_,\C)_0$ & $(q_4,q_5,q_6)$ &\\
$(1_i,0_\_,D)_0$ & $(q_4,q_5,q_6)$ & $(1_i,1_0,D)_0$ & $(q_4,q_5,q_6)$ &\\
$(1_i,1_\_,D)_0$ & $(q_4,q_5,q_6)$ & $(1_i,0_0,D'^p)_0$ & $(q_4,q_5,q_6)$ &\\
$(0_i,1_\_,D)_0$ & $(q_4,q_5,q_6)$ & $(0_i,0_0,D^p)_0$ & $(q_4,q_5,q_6)$ &\\
$(q_{1,i},q_{2,j},D)_0$ & $(a,a,q_6)_{j+1}$ &
                                              $(q_{1,i},q_{2,j+1},D)_0$
                                                       &
                                                         $(a,a,q_6)_{j+1}$
                                                                         &  \begin{minipage}{2.5cm}
                                                                           with
                                                                           $a\in\{0,1\}$
                                                                           and
                                                                           $j<i$-1
                                                                           \end{minipage}\\
$(q_{1,i},q_{2,j},D)_0$ & $(a,$1-$a,q_6)_{j+1}$ & $(q_{1,i},q_{2,\_},\C)_0$ & $(a,$1-$a,q_6)_{j+1}$ & with $a\in\{0,1\}$\\
$(q_{1,i+1},q_{2,i},D)_0$ & $(a,a,q_6)_{i+1}$ & $(q_{1,i+1},q_{2,\_},\RL)_0$ & $(a,a,q_6)_{i+1}$ & with $a\in\{0,1\}$\\
$(q_{1,i},q_{2,j},D'^p)_0$ & $(q_4,0,q_6)_{j+1}$ & $(q_{1,i},q_{2,\_},\C)_0$ & $(q_4,1,q_6)_{j+1}$ & \\
$(q_{1,i},q_{2,j},D'^p)_0$ & $(q_4,1,q_6)_{j+1}$ & $(q_{1,i},q_{2,j+1},D'^p)_0$ & $(q_4,0,q_6)_{j+1}$ & \\
$(q_{1,i},q_{2,j},D^p)_0$ & $(0,0,q_6)_{j+1}$ & $(q_{1,i},q_{2,\_},\C)_0$ & $(0,1,q_6)_{j+1}$ & \\
$(q_{1,i},q_{2,j},D^p)_0$ & $(1,0,q_6)_{j+1}$ & $(q_{1,i},q_{2,j+1},D)_0$ & $(1,1,q_6)_{j+1}$ & with $j<i$-1\\ 
$(q_{1,i},q_{2,j},D^p)_0$ & $(0,1,q_6)_{j+1}$ & $(q_{1,i},q_{2,j+1},D^p)_0$ & $(0,0,q_6)_{j+1}$ & \\
$(q_{1,i},q_{2,j},D^p)_0$ & $(1,1,q_6)_{j+1}$ & $(q_{1,i},q_{2,j+1},D'^p)_0$ & $(1,0,q_6)_{j+1}$ & \\ 
$(q_{1,i+1},q_{2,i},D^p)_0$ & $(1,0,q_6)_{i+1}$ & $(q_{1,i+1},q_{2,\_},\RL)_0$ & $(1,1,q_6)_{j+1}$ &\\
\end{tabular}
}
\end{center}

To prove that this protocol will succeed, we know by fairness that the
counting protocol will finish at some point. There will be several agents in state $W$.
By fairness, the leader will have seen all of them at some point.

Let consider a configuration appearing infinitely often.
We will show that the reset point can be reached from it
(and then will be reached by fairness).
\begin{itemize}
\item If the leader has its third state equal to $\A$, then we have it met all the agents to turn them
into $\SSS$. Then, we go into state $B$ in a
population where all other agents are in state $\SSS$.
\item If the leader is in state $\B$, it can only do an interaction one
after another: clearing the chain by following the leaders one identifier
after another.
We then
reach $\C$, keeping the number of $\SSS$ in the population.
\item If the leader is in state $\C$,
we have two cases:
\begin{itemize}
\item All agents are in state $\SSS$. It can be the case only if the leader's last
operation was becoming a $\C$ from a $\B$,
hence the second counter is equal to $1$. Then, if the leader repeats
the actions (turn an agent from $\SSS$ to $R$ + increment the counter),
at some point the two counters will match, and the leader will reach the state $\RL$.
\item At least one agent is not in state $\SSS$. We have the leader interacts with
it to go in state $\A$.
\end{itemize}
\item If the leader is in state $\{D,D^p,D'^p\}$, we can finish Step 5. The
leader will then be in state $\C$ or $\RL$.
\item If the leader is in state $\RL$, then the two counters must be equal.
Hence, the leader turned exactly the right number of agents from state $\SSS$ to $R$,
the reset is performed.
\end{itemize}
The population will reach at some point the desired configuration. \hfill \qed
 \end{proof}

To run the computations of the following protocols, we first reach the end
of this protocol. More precisely, the leader will start the next steps after
having reached the state $\RL$. If at some point the leader replaces
the $\RL$ with a $\A$ using the reset protocol, then it knows it has to restart
the next steps. While the leader will turn into $\SSS$ all the other agents,
it will also reset the computations elements from the next steps.

The last time the leader will turn into $\RL$ and  will know that all
agents are ready to start the next steps. It also knows that the size of the population
is encoded correctly on the chain.

%

This protocol ensures that the leader has access to all the input
at some point. With population protocols, the leader will never be able to know
if it has turned every agent into some state $\SSS$, because it has not the
ability to know the size of the population.


\subsection{Access to the Input} 

We now
introduce a protocol that computes the number of agents
that had the input $s\in\Sigma$ and were given the identifier $Id$.

This is a sub-protocol that can be used by the
main protocol at any moment. Hence, the information cannot be definitively kept 
and has to be precisely given  at the requested moment.
We cannot accept any error as it could not be detected on time to correct the computation.


\begin{proposition}\label{prop:sid}
If we have $f(n)=\Omega(\log n)$ identifiers and 
if the reset protocol has finished,
for all input $s\in \Sigma$ and for all $Id\le f(n)$,
there exists a protocol that writes on the chain the number of agents initialized
as $s_{Id}$.  
\end{proposition}

\begin{proof}
We assume that the population is already reseted to state $R$ using a Reset Protocol.

We do not give here the formal description of protocol, only a description of its process:
\begin{enumerate}
\item[0.] The agents will have a state corresponding to a $4$-tuples:
  \begin{itemize}
  \item the first element of the $4$-tuple is assumed to implement the
    counting protocol of previous subsection. We assume that the reset
    protocol has finished, so the counting protocol on this
first element is over.
\item the second element of the $4$-tuple will be used to count the agents
  with input $s_{Id}$.
\item the third element of the $4$-tuple will be used to implement another counting protocol, similar to the one of previous section. This will be used to recount the population to check that
    every agent has met the leader since the beginning of this process
    (by checking that the value encoded by first elements is equal to
    the value encoded by third elements).
\item the last fourth element of the $4$-tuple is
    here to determine whether the agent has already been  counted by the leader
    yet (for the counters corresponding to the second and third
    element of the 4-tuple).  It  is in
    $\{Y,N\}$ and is equal to $Y$ if and only if the leader already
    counted it.
  \end{itemize}

\item The leader looks for an agent it has not recounted again (i.e. with its 4th state equals to $N$).
When it meets one, it switches this agent's internal state from $N$ to $Y$, and 
it looks if its input was $s_{Id}$ or not. If it is, then it increments the second and the third counter,
otherwise it only increments  the third.
\item The leader then looks if the first and the third counter are equal. If not, it goes back to 
step $1$, if yes the computation is over.
\end{enumerate}

Since the counting protocol is over (if not, the population will be reseted again and again
until the counting is over), the size is known. With that, we are sure to have counted each agent
started in state $s_{Id}$, as the leader must have seen each agent in this protocol before
finishing it.\hfill \qed

%
%
 \end{proof}

\begin{remark}
  In other words, if at some moment, the population needs to know
the number of agents which started in the state $s_{Id}$, then this is
possible.
\end{remark}

\subsection{Turing Machine Simulation}\label{sec:def}

With all these ingredients we will now be able to access to the input
easily. We will also use the chain to simulate a tape of a
Turing Machine. 

The result obtained in this section  is a weaker bound than the one we will obtain latter. The principle of this proof helps to understand the stronger result.

\begin{proposition}
Any language in $MNSPACE(\log n, \log n)$ can be recognized by an homonym population
protocol with $\log n$ identifiers.
\end{proposition}

\begin{proof}
  The main key of this proof is to use the chain 
  as a tape for a Turing Machine. To simulate the tape of
  the Turing Machine, we store the position where the head of the
  Turing machine is by memorizing on which multiset the head is (via
  the corresponding identifier) and its relative position inside this
  multiset: the previous protocol will be used to find out the number of
  agents with some input symbol in the current multiset, in order to
  update all these information and simulate the evolution
  of the Turing Machine  step by step.

More precisely, let $M\in MNSPACE(\log n, \log n)$. There exists some $k\in\mathbb{N}$ such that $M$ uses at most $k\log n$ bits
for each input of size $n$.
 To an input $x_1$\#$x_2$\#$\ldots$\#$x_{f(n)}$, we associate the input configuration
with, for each $s\in\Sigma$ and for each $i\le f(n)$, $|x_i|_{s}$ agents in state $k$ with the identifier $(i-1)$,
 $|x_i|_{s}$ being the number of $s$ in $x_i$.

The idea is to use the chain as the tape of the Turing Machine. We
give $k$ bits to each agent,  so that
the protocol  has a tape of the correct length (the chain is of size $\log n$).
We just need to simulate the reading of the input (the writing
will be intuitively be performed by the leader keeping track of the identifier
of the agent where the machine's head is).
The protocol starts by counting the population and resetting agents after that.

We assume that symbols on $\Sigma$ are ordered. 
Since the language recognized by $M$ is $\log n$-symmetric, we can reorganize the input by permuting the $x_i$'s
such that the input symbols are ordered (i.e. $\Sigma=\{s_1,s_2,\ldots\}$ and 
$x_i=s_1s_1\ldots s_1s_2\ldots$).

Here are the steps that perform the simulation of reading the tape: 
\begin{enumerate}
\item[0.] The chain contains two counters. The leader also stores an identifier $Id$ and a state $s$.
The first counter stores the total of $s_{Id}$ computed at some point by the protocol of
Proposition~ \ref{prop:sid}. The second counter $c_2$ is the position the reading head.
The simulated head is on the $c_2$th $s$ of $x_{Id+1}$.
\item At the beginning of the protocol, the population counts the number of agents with input $s_1$
and identifier $0$, where $s_1$ is the minimal
element of $\Sigma$. $c_2$ is initialized to $1$.
\item When the machine needs to go to the right on the reading tape, $c_2$
  is incremented. If $c_2$ equals $c_1$, then
the protocol looks for the next state $s'$ in the order of $\Sigma$, and count the number of $s'_{Id}$. If this
value is $0$, then it takes the next one. If $s$ was the last one, then the reading tape will consider to be on a \#.

If the reading head was on a \#, then it looks for the successor identifier of $Id$, and counts the number of $s_1$.
If $Id$ was maximal, then the machine knows it has reached the end of the input tape.

\item The left movement process is similar to this one.
\end{enumerate}

This protocol can simulate the writing on a tape and the reading of the input.\\

To simulate the non deterministic
part, each time the leader needs to make a non deterministic choice
between two possibilities, it looks for an agent. If the first
agent the leader meets has its identifier equal to $1$, then the leader makes
the first choice, otherwise it makes the second one.
When the computation is over, if it rejects, it reset the simulation
and starts a new one.

By fairness, if there is a path of non deterministic choices for the machine,
the protocol will at some point use it and accept the input, as would do $M$.
If not, as all the agents will stay in a rejecting state, the protocol will reject
the input.

This protocol simulates $M$.\hfill \qed

\end{proof}

\begin{corollary}
Let $f$ such that $f(n)=\Omega(\log n)$. 

Any language in $MNSPACE(f(n),$ $f(n))$ can be recognized by an homonym population protocol with $f(n)$ identifiers.
\end{corollary}
\begin{proof}
We use the same protocol (which is possible as the size of the population can be computed).
Since the chain of identifiers has a length of $f(n)$, we have access to a tape of size $f(n)$.\hfill \qed

\end{proof}

\subsection{Polylogarithmic Space}

We prove now the exact characterization of what can be computed
by our model: functions computable by Turing Machines on
polylogarithmic space. To prove it, we first prove several propositions.
The combination of the three following results permit to conclude the main theorem.

\begin{proposition}\label{prop:truc}
Let $f$ such that $f=\Omega(\log n)$.
Let $k$ be a positive integer.

Any language in $MNSPACE\left(\log^kn,f(n)\right)$ can be recognized by a protocol
with $f(n)$ identifiers.
\end{proposition}

\begin{proof}
The idea here is that, by combining several identifiers together, we get much more identifiers available, increasing
the chain and space of computation: Indeed, if we combine $m$ {\id}s together
in a $m$-tuple, then we get $f(n)^m$ possible identifiers. The basic idea is to
count in base $f(n)$: the leader finds $f(n)^m$ agents and
distribute to each of them a unique new identifier (encoding as a $m$-tuple of
original identifiers). 

To do so, first the population performs the computation of the size of
the population. This also builds a chain of
all the identifiers. Then, the leader creates a counter of $m$ identifiers, initialized at $(0,0,\dots,0)$ (seen as the number
$0\ldots0$ written in base $f(n)$). It looks for an agent in state $N$
(i.e. $N$ corresponds to the state of a agent which has not been given
a $m$-tuple identifier yet) and transmits its new identifier: that is the
current $m$-tuple stored in the leader. The leader then increments his
counter by $1$. As soon as it has finished (by giving $f(n)^m$ or $n$ identifiers, depending on what happens first),
the protocol can then work on a tape of space $f(n)^m$.

Since $f(n)=\Omega(\log n)$, there exists some $m$ such that $f(n)^m\ge \log^kn$.\hfill \qed

 \end{proof}

\begin{proposition}
Let $f$ such that there exists some real $r>0$ such that we have $f(n)=\Omega(\log^r n)$.

Any language in $\bigcup_{k\ge1}MNSPACE(\log^k n,f(n))$ can be
recognized by an homonym population protocol with $f(n)$ identifiers.
\end{proposition}
\begin{proof} 
We only need to treat the counting protocol when $r<1$ (the case $r=1$ is treated in Proposition~\ref{prop:truc},
the case $r>1$ is a direct corollary of this proposition).

In previous constructions, to count the population, we needed at least
$\log n$ identifiers. The idea is to use $\ell$-tuples to encode
identifiers, when the number of identifiers is too low. 

By taking  $\ell=\lceil\frac 1r\rceil$ we have $f(n)^\ell=\Omega(\log n)$,
and a counting protocol can be implemented after distributing these
new identifiers (using a process similar to previous proposition). 

More precisely, in the counting protocol, when agents realize that $f(n)$ might be reached and they need more identifiers, they use the tuple,
storing the maximal \id\ $Id_1$. If at some point, they realize that a higher \id\ $Id_2$ exists,
they just do an appropriate update of the numbers stored in the chain.\hfill \qed

  \end{proof}









\begin{proposition}\label{th:olivier}
Consider a predicate computed by a protocol with $f(n)$
identifiers. Assume that $f(n)=O(\log^\ell n)$ for some $\ell  \ge 1$.   

The predicate is in $MNSPACE(\log^k n,f(n))$ for some positive integer $k$. 
\end{proposition}
\begin{proof}
We need to prove that there exists a Turing Machine that can compute, for any given input $x$,
the output of protocol $P$.

From definitions, given some input $x$, $P$ outputs the output $y$ on
input $x$ if and only if  there exists a finite sequence $(C_i)_{i\in\mathbb{N}}$, starting from an
initial condition $C_0$ representing $x$, that reaches at some finite time
$j$  some configuration $C_j$ with interpretation $y$, and so that any
configuration reachable from $C_j$ has also interpretation $y$. 

This latter property can be expressed as a property on the graph of
configurations of the protocol, i.e. on the graph whose nodes are
configurations of $n$ agents, and whose edges corresponds to unique step
reachability:  one must check the existence of a path
from $C_0$ to some $C_j$ with interpretation $y$ so that there is no
path from $C_j$ to some other $C'$ with interpretation different from
$y$. 

Such a problem can be solved in $NSPACE(\log N)$ where $N$
denotes the number of nodes of this graph of configurations. 
Indeed, guessing a path from $C_0$ to
some $C_j$ can easily be done in $NSPACE(\log N)$ by guessing
intermediate nodes (corresponding to configurations) between $C_0$ and $C_j$. There remains to
see that testing if there is no
path from $C_j$ to some other $C'$ with interpretation different from
$y$ can also be done in $NSPACE(\log N)$ to conclude.

But observe that testing if there is a path from $C_j$ to some other
$C'$ with interpretation different from $y$ is clearly in $NSPACE(\log
N)$ by guessing $C'$. From Immerman-Szelepcsnyi's Theorem
\cite{immerman1988nondeterministic,Szelepcsenyi} we know that one has 
$NSPACE(\log N)=co-NSPACE(\log N)$. Hence, testing if there is no path
from $C_j$ to some other $C'$ with interpretation different from
$y$ is indeed also in $NSPACE(\log N)$.

It remains now to evaluate $N$: 
For a given identifier $i$, an agent encodes basically some basic state $b
\in B$, and $d$ identifiers $u_1,u_2,\dots,u_d$. There are at most $n$ agents
in  a given state $(i,b,u_1,u_2,\dots,u_d)$.
Hence $N=O(n^{|B|\cdot f(n)^{d+1}}).$ 
In other words, the algorithm above in $NSPACE(\log N)$ is
hence basically in $MNSPACE((|B| \cdot f(n)^{d+1}) \log n,f(n))$. As a
consequence, this
is 
included in $MNSPACE(\log^k n,f(n))$
for some $k$. \hfill \qed

\end{proof}

\begin{theorem}\label{th:main}
Let $f$  such that for some $r$, we have $f(n)=\Omega(\log^r n)$.
The set of functions computable by homonym population protocols with $f(n)$
identifiers corresponds exactly to $\bigcup_{k\ge1}MNSPACE(\log^kn,f(n))$.

\end{theorem}

\section{The Rest of the Hierarchy}\label{resthirera}

\subsection{Population with $n^\epsilon$ Identifiers}

One can go from  $n^\epsilon$ (with $\epsilon >0$) to a space of
computation equivalent to the case where $f(n)=n$:  We just need to use a $k$-tuple of identifiers,
as in Proposition~\ref{prop:truc}.


\begin{theorem}\label{th:sqrt}
Let $f$ such that there exists some $k\in \mathbb{N}$ such that $f(n)\ge n^{1/k}$.
The set of functions computable by homonym population protocols with $f(n)$
identifiers corresponds exactly to $MNSPACE(n\log n,f(n))$.
\end{theorem}



\begin{remark}
This result does not need the two restrictions of knowing if an identifier is equal to $0$ or if
two identifiers are consecutive. The result holds when the set of possible identifiers
$U$ is chosen arbitrarily and when the
restrictions over the rules are those in \cite{guerraoui2009names}.
\end{remark}

\subsection{Population with $o(\log n)$ Identifiers}

This time, we consider we have really few identifiers.
To write the size of the population in binary, we need to differentiate
$\log n$ agents. With $o(\log n)$ identifiers, it is no longer possible.
Because of that, the counting protocol and the reset protocol can no longer
be used to simulate Turing Machines.

In this section, we consider two cases: $f(n)=o(\log n)$ and
a constant number of identifiers.
We have a characterization when
this number is constant: it leads to population protocols. In the general case,
the population is more powerful, but we do not have any exact characterization.


\begin{theorem}\label{hppconst}
Let $f$ such that for some $k\in\N$, we have, for all $n$, $f(n)\le k$.

The set of functions computable by homonym population protocols with $f(n)$
identifiers corresponds to the semilinear sets over $\Sigma\times[0,k-1]$.
\end{theorem}
\begin{proof}
Each agent tries to put in its internal state the value of its identifier 
by using the following algorithm:
\begin{itemize}
\item If the identifier is $0$, then the agent puts $0$.
\item If an agent knowing its identifier $id$ meets an agent with identifier $Next(id)$, then
the second knows its identifier by incrementing $id$.
\end{itemize}
From this,  it is possible to compute semilinear predicates.
Indeed, agents know now their initial state and can run the corresponding population protocol.\\

The proof that only semilinear sets can be computed is quite simple: we see a community protocol
as a population protocol where the identifiers are directly included in the states.
More precisely, $Q'=B\times[1,k]\times([1,k]\cup\{\_\})^{d+1}$, and $\delta'=\delta$.\hfill \qed

\end{proof}

\begin{proposition}\label{contrex1}
There exists some non semilinear predicates over $\Sigma\times[1,k]$
and some function $f=o(\log n)$ such that there exists
a homonym population protocols with $f(n)$ identifiers that computes this predicate.
\end{proposition}
\begin{proof}
For example, it is possible to compute the predicate asking if the
 agents with an even identifier are in majority.

To compute this it suffices to determine for each agent if its identifier
is even or odd. The way to compute it is as follows:
\begin{itemize}
\item If the identifier is $0$, then the agent remembers its identifier is odd.
\item If an agent knowing its parity meets an agent with an identifier
that is the directly next one, then
the second knows its parity by switching the other's one.
\end{itemize}
We then run the protocol $[x_{\text{even identifier}}>x_{\text{odd identifier}}]$.\hfill \qed

\end{proof}

\begin{remark}\label{contrex2}
Another counter-example is that we can compute the following predicate
 $$\left[\sum\limits_{id\le id_{max}/2} a_ix_{s_i,id}
-\sum\limits_{id> id_{max}/2} b_ix_{s_i,id}\ge c\right]$$
where $x_{s_i,id}$ is the number of agents with input $s_i$ and identifier $id$.
This predicate corresponds to a threshold predicate when we take into consideration
whether the identifier of the agent is in the first or second half of the present ones.

This protocol computes the threshold predicate with first value $b_i$ for each agent
with non $0$ identifier and input $s_i$, and $a_i$ for the agents with identifier $0$ and
input $s_i$.

To compute the medium identifier, we have $d=2$.
For an agent $q_{i,j,k}$, $i$ means agent's own identifier, $j$ is the medium candidate and,
$k$ represents $2j$ or $2j+1$. At the beginning, if the identifier is $0$, then we have $q_{0,0,0}$,
otherwise, we get $q_{i,\_,\_}$. The state $q$ means "I need to increment $k$ if $k+1$
is present". The state $q^{++}$ means "I need to increment both $j$ and $k$ if
$k+1$ is present". The state $q^{+}$ means "I need to increment  $j$ if
$k+1$ is present". Hence, each time we increment twice $k$, we increment once $j$.

The rules are:
\begin{center}
\begin{tabular}{ r @{\hspace{0,2cm}} l @{$\rightarrow$} r @{\hspace{0,2cm}}  l l}
$q_{i,\_,\_}$ & $q'_{0,j',k'}$ & $q_{i,0,0}$ & $q'_{0,j',k'}$ & \\
$q_{i,j,k}$ & $q'_{k+1,j',k'}$ & $q^{++}_{i,j,k+1}$ & $q'_{k+1,j',k'}$ & \\
$q^{++}_{i,j,k}$ & $q'_{k+1,j',k'}$ & $q^+_{i,j,k+1}$ & $q'_{k+1,j',k'}$ & \\
$q^{+}_{i,j,k}$ & $q'_{j+1,j',k'}$ & $q_{i,j+1,k}$ & $q'_{j+1,j',k'}$ & \\
\end{tabular}
\end{center}

As soon as an agent in input $s_i$ realizes its identifier is smaller or equal to its $j$,
it adds $a_i-b_i$ to its state if possible (otherwise, it waits an occasion to add it
to another agent).

By fairness, all agents will determine at some point if their identifier is greater
or smaller to half the highest one, and then the leader will be able to compute the
right output.
\end{remark}

\section{Passively Mobile Machines}
\label{sec:4demi}

We now show how previous constructions
 improve the results about the passively mobile protocol model  \cite{chatzigiannakis2011passively}.
This section treats the case where $S(n)=O(\log\log n)$ in the
passively mobile protocol model. Table \ref{tab:Chatzigiannakis} gives a summary of this hierarchy.
$PMSPACE(f(n))$ corresponds to the class of languages recognized
by Passively Mobile Agents using space $O(f(n))$.





\begin{theorem}\label{th:grec}
%
$PMSPACE(\log\log n)=\bigcup_{k\ge1}SNSPACE(log^kn)$.\end{theorem}
\begin{proof}
\textbf{1.} $\bigcup_{k\ge1}SNSPACE(log^kn)\subset PMSPACE(\log\log n)$.

The idea of this proof is quite simple: Let $M\in SNSPACE(log^kn)$.
We can notice that $SNSPACE(log^kn)\subset MNSPACE(log^k n, \log n)$ 
(as the space of computation is the same and symmetric is equivalent
to be a single multiset).
From Theorem \ref{th:main}, there is a population protocol computing M. We will simulate it.
With space $O(\log\log n)$, we can simulate
a population protocol with $O(2^{\log\log n})=O(\log n)$ identifiers.

Indeed, we adapt a bit the counting protocol.
 At the beginning, each agent has identifier $0$ in order to create $\log n$ identifiers.
When two agents with the same \id\ meet, if each one contains the integer $1$, then
the first switch its integer to $0$ and, the other increases its own \id.

We then just need to simulate the behavior of each agent as if they have started 
with their created \id. It requires a space of size $|B|+(d+1)\log\log n$ plus some constant,
which is enough.

\textbf{2.} $PMSPACE(\log\log n)\subset\bigcup_{k\ge1}SNSPACE(log^kn)$:
The proof is similar to the one of Theorem \ref{th:olivier}:
It is a question of accessibility in the  graph of the configurations.
We need to compute the number of possible configurations denoted by $N$.

For each agent, there are $|Q|$ possible states and 4 tapes of length
$\alpha\log\log n$ for some $\alpha$. Hence, there are $|Q|\times|\Gamma|^{4\alpha\log\log n}$
possible states for each agent.

Now $|\Gamma|^{4\alpha\log\log n}=|\Gamma|^{\log(\log^{4\alpha}n)}=\left( \log^{4\alpha}n\right)^{\log |\Gamma|}$

For each possible state, there are at most $n$ agents sharing it.

Hence, $N=O\left(n^{|Q|\times\left( \log^{4\alpha}n\right)^{\log |\Gamma|}}\right)$.

The accessibility can be computed by a machine in space complexity
$O(\log N)$, which means a space $O\left(|Q|\times\left( \log^{4\alpha}n\right)^{\log |\Gamma|}\log n\right)
=O(\log^kn)$ for some $k\in\N$.\hfill \qed

 \end{proof}

With a similar proof, we can get the following result that gives a good clue for the gap
between $\log\log n$ and $\log n$:

\begin{corollary}
Let $f$ such that $f(n)=\Omega(\log\log n)$ and $f(n)=o(\log n)$.

$$SNSPACE(2^{f(n)}f(n))\subset PMSPACE(f(n))\subset SNSPACE(2^{f(n)}\log n).$$
\end{corollary}

\section{Summary} \label{sec:5}

From the model given  by Guerraoui and Ruppert \cite{guerraoui2009names}, we introduced a hierarchy according
to the number of distinct identifiers in the population:
\begin{itemize}
\item  The existence of identifiers is useless with a constant number of the identifiers.
\item Homonym population protocols with $\Theta(\log^r n)$ identifiers
  can exactly  recognize any language in $\bigcup_{k \in \mathbb{N}} MNSPACE\left(\log^k n
\right)$.
\item Homonym population protocols with $\Theta(n^\epsilon)$ identifiers
  have same power that homonym population protocols
  with $n$ identifiers.
\end{itemize}

It remains an open and natural question: is the knowledge of consecutive values of two identifiers
 crucial or not?  Our guess is that this knowledge  is essential to compute  
the size of the population. Protocols without this assumption have not been found yet.

Chatzigiannakis \emph{et al.}  \cite{chatzigiannakis2011passively} started a hierarchy over protocols
depending on how much space of computation each agent has.
The paper left an open question on the gap between $o(\log\log n)$ and $O(\log n)$.
We provided an answer, by proving that with $\Theta(\log\log n)$ space,
exactly $\bigcup_{k \in \mathbb{N}} SNSPACE\left(\log^k n
\right)$ is computed. However, it remains the gap between $O(\log\log n)$ and $O(\log n)$, where we 
currently just have the following bounds: 

$SNSPACE(2^{f(n)}f(n))\subset PMSPACE(f(n))\subset SNSPACE(2^{f(n)}\log n)$.


\end{document}